\documentclass{eptcs}
\providecommand{\event}{VPT 2017} 

\usepackage{stmaryrd}
\usepackage{latexsym}
\usepackage{amsmath}
\usepackage{amssymb}
\usepackage{wasysym}
\usepackage{proof}
\usepackage{tikz}

\newcommand{\expr}[1]{#1{}}

\newcommand{\var}[2]{\mathit{#1}#2}
\newcommand{\abs}[3]{\lambda #1{.#2{#3}}}
\newcommand{\app}[3]{#1{~#2{#3}}}

\newcommand{\args}[3]{#1{\ldots#2{#3}}}

\newcommand{\casedots}[6]{\mathsf{CASE}~#1{~\mathsf{of}~#2{\rightarrow #3{~| 
\cdots|~#4{\rightarrow #5{#6}}}}}}

\newcommand{\Where}[6]{#1{} ~ \mathsf{WHERE} ~ #2 = #3 \ldots #4 = #5{#6}}

\newcommand{\ignore}[1]{}

\newenvironment{proof}{{\em Proof}.~}{\hfill $\Box$\\}

\newtheorem{theorem}{Theorem}[section]

\newtheorem{lemma}[theorem]{Lemma}
\newtheorem{definition}[theorem]{Definition}  
\newtheorem{example}{Example}

\title{Generating Loop Invariants for Program Verification by Transformation}

\author{G.W. Hamilton
\institute{School of Computing\\
Dublin City University\\
Ireland}
\email{hamilton@computing.dcu.ie}}

\def\titlerunning{Generating Loop Invariants}

\def\authorrunning{G.W. Hamilton}

\begin{document}

\maketitle

\begin{abstract}
Loop invariants play a central role in the verification of imperative programs. However, finding these invariants is often a difficult and time-consuming
task for the programmer. We have previously shown how program transformation can be used to facilitate the verification of functional programs, 
but the verification of imperative programs is more challenging due to the need to discover these loop invariants.
In this paper, we describe a technique for automatically discovering loop invariants. Our approach is similar to the 
induction-iteration method, but avoids the potentially exponential blow-up in clauses that can result when using this and other methods. 
Our approach makes use of the {\em distillation} program transformation algorithm to transform clauses into a simplified form that facilitates the identification
of similarities and differences between them and thus help discover invariants. We prove that our technique terminates, and demonstrate its successful
application to example programs that have proven to be problematic using other approaches. We also characterise the situations where our technique
fails to find an invariant, and show how this can be ameliorated to a certain extent.
\end{abstract}

\section{Introduction}

The verification of imperative programs generally involves annotating programs with {\em assertions}, and then using a
theorem prover to check these annotations. Central to this annotation process is the use of {\em loop invariants} which are
assertions that are true before and after each iteration of a loop. However, finding these invariants is a difficult and 
time-consuming task for the programmer, and they are often reluctant to do this. In previous work \cite{HAMILTON07B}, 
we have shown how to make use of program transformation in the verification of functional programs.
However, the verification of imperative programs is not so straightforward due to the need to discover these invariants prior to verification.
In this paper, we describe a technique for automatically discovering loop invariants, thus relieving the programmer of this burden. 
Our technique relies upon the programmer having provided a {\em postcondition} for the program; this is much less onerous than providing 
loop invariants as it generally forms part of the specification of the program.

The technique we describe is similar to the induction-iteration method of Suzuki and Ishihata \cite{SUZUKI77}, but we
overcome the problems associated with that method which were potential non-termination and exponential blow-up in
the size of clauses. Similarly to the induction-iteration method, our technique involves working backward through the
iterations of a loop and determining the assertions that are true before each iteration.  
We use the {\em distillation} program transformation \cite{HAMILTON07A,HAMILTON12} 
to transform assertions into a simplified form that facilitates the identification of similarities 
and differences between them. Commonalities between these assertions are identified, and they are generalised 
accordingly to give a putative loop invariant that can then be verified. We prove that our technique terminates
and demonstrate its successful application to example programs that have proven to be problematic using other approaches.
We also characterise the situations where our technique fails to find an invariant.

The remainder of this paper is structured as follows. In Section 2, we describe the simple imperative language that will be
used throughout the paper. In Section 3, we provide some background on the use of loop invariants in the verification of programs 
written in this language. In Section 4, we give a brief overview of the distillation program transformation algorithm which is used
to simplify assertions in our approach. In Section 5, we describe our technique for the automatic generation of loop invariants
and prove that it terminates. In Section 6, we give a number of examples of the application of our technique, demonstrating
where it succeeds on problematic examples and where it does not. In Section 7, we consider related work and compare these to our 
own our technique. Section 8 concludes and considers future work.

\section{Language}

In this section, we introduce our object language, which is a simple imperative programming language.
\begin{definition}[Language Syntax]
\normalfont{The syntax of our object language is as shown in Figure \ref{syntax}.
\begin{center}
\begin{figure}[htb]
\begin{center}
\begin{tabular}{lrll}
$S$ & $::=$ & $\mathsf{SKIP}$ & Do nothing \\
& $|$ & $V$ $:=$ $E$ & Assignment\\
& $|$ & $S_1$ $;$ $S_2$ & Sequence \\
& $|$ & $\mathsf{IF}$ $B$ $\mathsf{THEN}$ $S_1$ $\mathsf{ELSE}$ $S_2$ & Conditional \\
& $|$ & $\mathsf{BEGIN}$ $\mathsf{VAR}$ $V_1 \ldots V_n$ $S$ $\mathsf{END}$ & Local block \\
& $|$ & $\mathsf{WHILE}$ $B$ $\mathsf{DO}$ $S$ & While loop 
\end{tabular}
\end{center}
\caption{Language Syntax}
\label{syntax}
\end{figure}
\end{center}
$E$ corresponds to natural number expressions which belong to the following datatype:
$$Nat ::= \mathsf{Zero}~|~\mathsf{Succ}~Nat$$
We use the shorthand notation $0,1,\ldots$ for $\mathsf{Zero}, \mathsf{Succ~Zero}, \ldots$ \\
\\
$B$ corresponds to boolean expressions which belong to the following datatype:
$$Bool ::= \mathsf{True}~|~\mathsf{False}$$}
\end{definition}
These expressions are defined in a simple functional language with the following syntax.
\begin{definition}[Expression Syntax]
\normalfont{The syntax of expressions in our language is as shown in Figure \ref{grammar}.}
\end{definition}
\begin{figure}[htb]
\begin{center}
\begin{tabular}{lrll}
$\expr{\var{E}}$ & ::= & $\expr{\var{V}}$ & Variable \\
& $|$ & $\expr{\app{\var{C}}{\args{\var{E_1}}{\var{E_k}}}}$ & Constructor Application \\
& $|$ & $\expr{\abs{\var{V}}{\var{E}}}$ & $\lambda$-Abstraction \\
& $|$ & $\expr{\var{F}}$ & Function Call \\
& $|$ & $\expr{\app{\var{E_0}}{\var{E_1}}}$ & Application \\
& $|$ & $\expr{\casedots{\var{E_0}}{\var{P_1}}{\var{E_1}}{\var{P_k}}{\var{E_k}}}$ & Case Expression \\ 
& $|$ & $\expr{\Where{\var{E_0}}{\var{F_1}}{\var{E_1}}{\var{F_n}}{\var{E_n}}}$ & Local Function Definitions \\
\\
$\expr{\var{P}}$ & ::= & $\expr{\app{\var{C}}{\args{\var{V_1}}{\var{V_k}}}}$ & Pattern
\end{tabular}
\end{center}
\caption{Expression Syntax}
\label{grammar}
\end{figure} 
An expression can be a variable, constructor application, $\lambda$-abstraction, function call, application, $\mathsf{CASE}$ or $\mathsf{WHERE}$. 
Variables introduced by $\lambda$-abstractions, $\mathsf{CASE}$ patterns and $\mathsf{WHERE}$ definitions are {\em bound}; all other variables 
are {\em free}. We use $fv(E)$ to denote the free variables of $E$ and write $E \equiv E'$ if $E$ and $E'$ differ only in the names of bound variables.

The constructors are those specified above $(\mathsf{Zero},\mathsf{Succ},\mathsf{True},\mathsf{False})$.
We assume a number of pre-defined operators written in this language; the explicit definitions of these operators are unfolded as part of our transformation
rather than appealing to properties such as associativity, etc. For natural number expressions the operators $(+,-,*,/,\%,^\wedge)$ implement natural 
number addition, subtraction, multiplication, division, modulus and exponentiation respectively.
For boolean expressions the operators $(\wedge,\vee,\neg,\Rightarrow)$ implement conjunction, disjunction, negation and implication respectively. 
The relational operators $(<,>,\leq,\geq,=,\neq)$ are also defined.
\begin{definition}[Substitution]
\normalfont{ $\theta = \{V_1 \mapsto E_1, \ldots, V_n \mapsto E_n\}$ denotes a {\em substitution}.
If $E$ is an expression, then $E\theta = E\{V_1 \mapsto E_1, \ldots, V_n \mapsto E_n\}$ is the result of simultaneously 
substituting the expressions $E_1,\ldots, E_n$ for the corresponding variables $V_1,\ldots,V_n$, respectively, 
in the expression $E$ while ensuring that bound variables are renamed appropriately to avoid name capture.}
\end{definition}
\ignore{
\begin{definition}[Instance]
\normalfont{We say that an expression $e_1$ is an instance of another expression $e_2$ if there is a substitution $\theta s.t. e_1 \equiv e_2 \theta$.}
\end{definition}
}
We reason about the behaviour of our imperative programming language using {\em Floyd-Hoare style logic} \cite{FLOYD67,HOARE69}.
Specifications in this logic take the form of a triple $\{P\}~S~\{Q\}$, where $P$ and $Q$ are boolean expressions that denote the pre- and post-conditions 
respectively for imperative program $S$ i.e. if $P$ is true, then after execution of $S$, $Q$ will be true. These are therefore {\em partial correctness}
specifications, and do not say anything about the termination of programs. 
\begin{definition}[Floyd-Hoare Logic]
\normalfont{The rules and axioms of Floyd-Hoare logic for our imperative language are as shown in Figure \ref{proofrules}.
\begin{center}
\begin{figure}[htb]
\begin{center}
\begin{tabular}{cc}
$\{P\} ~ \mathsf{SKIP} ~ \{P\}$ & $\{Q\{V:=E\}\} ~ V :=E ~ \{Q\}$ \\
\\
$\infer{\{P\} ~ S_1;S_2 ~ \{R\}}{\{P\} ~ S_1 ~ \{Q\}, ~~~~~~ \{Q\} ~ S_2 ~ \{R\}}$ & $\infer{\{P\} ~ \mathsf{IF} ~ B ~ \mathsf{THEN} ~ S_1 ~ \mathsf{ELSE} ~ S_2 ~ \{Q\}}{\{P \wedge B\} ~ S_1 ~ \{Q\}, ~~~~~~ \{P \wedge \neg B\} ~ S_2 ~ \{Q\}}$ \\
\\
$\infer{\{P\} ~ \mathsf{BEGIN} ~ \mathsf{VAR} ~ V_1 \ldots V_n ~ S ~ \mathsf{END} ~ \{Q\}}{\{P\} ~ S ~ \{Q\},~~~~~~ V_1 \ldots V_n \notin fv(P), fv(Q)}$ & $\infer{\{I\} ~ \mathsf{WHILE} ~ B ~ \mathsf{DO} ~ S ~ \{I \wedge \neg B\}}{\{I \wedge B\} ~ S ~ \{I\}}$ \\ 
\\
$\infer{\{P\} ~ S ~ \{Q\}}{P \Rightarrow P', ~~~~~~ \{P'\} ~ S ~ \{Q\}}$ & $\infer{\{P\} ~ S ~ \{Q\}}{\{P\} ~ S ~ \{Q'\}, ~~~~~~ Q' \Rightarrow Q}$
\end{tabular}
\end{center}
\caption{Floyd-Hoare Logic}
\label{proofrules}
\end{figure}
\end{center}
In the rule for the $\mathsf{WHILE}$ loop, the assertion $I$ is called the {\em loop invariant}.}
\end{definition}

\section{Loop Invariants}

A loop invariant is an assertion that is true before and after each iteration of the loop, and usually needs to be provided by the programmer. 
A loop which is annotated in this way is denoted by $\mathsf{WHILE} ~ B ~ \mathsf{DO} ~ \{I\} ~ S$
\begin{definition}[Requirements of Loop Invariants]
\normalfont{The three requirements of the invariant $I$ of the loop $\{P\}~\mathsf{WHILE} ~ B ~ \mathsf{DO} ~ \{I\} ~ S~\{Q\}$ are as follows:
\begin{enumerate}
\item $P \Rightarrow I$
\item $\{I \wedge B\}~S~\{I\}$
\item $(I \wedge \neg B) \Rightarrow Q$
\end{enumerate}
Thus, the precondition $P$ should establish the invariant before executing the loop, the loop body $S$ should maintain the invariant, 
and the invariant should be sufficient to establish the postcondition $Q$ after exiting the loop.}
\label{invariant}
\end{definition}
\begin{example}
\normalfont{Consider the program shown in Figure \ref{example1}.
\begin{center}
\begin{figure}[h]
\begin{center}
\begin{tabular}{l}
$~~~\{\mathsf{n \geq 0}\}$ \\
$~~~\mathsf{x:=0;}$ \\
$~~~\mathsf{y:=1;}$ \\
$~~~\mathsf{WHILE ~ x < n ~ DO}$ \\
$~~~~~~\mathsf{BEGIN}$ \\
$~~~~~~~~~ \mathsf{x:=x+1;}$ \\
$~~~~~~~~~ \mathsf{y:=y * k}$ \\
$~~~~~~ \mathsf{END}$ \\
$~~~\{\mathsf{y = k^\wedge n}\}$
\end{tabular}
\end{center}
\caption{Example Program}
\label{example1}
\end{figure}
\end{center}
This program calculates the exponentiation $\mathsf{k^\wedge n}$.
Say we wish to construct an invariant for the loop in this program which will allow it to be verified.
In \cite{FURIA10} it is observed that the required invariant is often a weakening of the postcondition for the loop
and can be obtained by mutating this postcondition. The assertion $\mathsf{y = k^\wedge x}$ is an invariant for this loop which
is a mutation of the postcondition. However, this invariant is not sufficient to allow verification of this program;
the additional invariant $\mathsf{x \leq n}$ is also required. In general, the problem of constructing appropriate invariants
which are sufficient to allow programs to be verified is undecidable. However, in this paper we show how we 
can automatically generate invariants which are sufficient to allow a wide range of programs to be verified.}
\end{example}
Our approach makes use of the {\em weakest liberal precondition} originally proposed by Dijkstra \cite{DIJKSTRA75}. 
\begin{definition}[Weakest Liberal Precondition]
\normalfont{We define the {\em weakest liberal precondition} for programs in our language, denoted as $WLP(S,Q)$, where $S$ is a program and $Q$ a 
postcondition. The condition $P = WLP(S,Q)$ if $Q$ is true after execution of $S$, and no condition weaker than $P$ satisfies this. 
The key difference of a weakest liberal precondition as opposed to a weakest precondition is that it does not say anything about the termination of programs. 
The rules for calculating $WLP(S,Q)$ for our programming language are as shown in Figure \ref{wlp}.
\begin{center}
\begin{figure}[htb]
\begin{center}
\begin{tabular}{rcl}
$WLP(\mathsf{SKIP},Q)$ & = & $Q$ \\
$WLP(V:=E,Q)$ & = & $Q\{V:=E\}$ \\
$WLP(S_1;S_2,Q)$ & = & $WLP(S_1,WLP(S_2,Q))$ \\
$WLP(\mathsf{IF}~B~\mathsf{THEN}~S_1~\mathsf{ELSE}~S_2,Q)$ & = & $(B \Rightarrow WLP(S_1,Q)) \wedge (\neg B \Rightarrow WLP(S_2,Q))$ \\
$WLP(\mathsf{BEGIN}$ $\mathsf{VAR}$ $V_1 \ldots V_n$ $S$ $\mathsf{END},Q)$ & = & $WLP(S,Q)$, where $V_1 \ldots V_n \notin fv(Q)$ \\
$WLP(\mathsf{WHILE}~B~\mathsf{DO}~\{I\}~S,Q)$ & = & $I \wedge ((B \wedge I) \Rightarrow WLP(S,I)) \wedge ((\neg B \wedge I) \Rightarrow Q)$ 
\end{tabular}
\end{center}
\caption{Weakest Liberal Precondition}
\label{wlp}
\end{figure}
\end{center}}
\end{definition}
Note that the weakest liberal precondition calculation for a loop requires that it has already been annotated with its invariant.
This implies that we should apply our techniques to inner loops first to determine their invariant before applying them to outer loops.

\section{Distillation}

The predicates produced in our approach are simplified using the {\em distillation} transformation \cite{HAMILTON07A,HAMILTON12}. 
Distillation is a fold/unfold program transformation that builds on top of positive supercompilation \cite{SORENSEN96}, but is more powerful, 
thus allowing more simplifications to be performed. The main distinguishing characteristic between the two algorithms is that in distillation, 
generalisation and folding are performed with respect to recursive terms, while in positive supercompilation they are not. In the work described 
here, we use distillation to transform predicates into a simplified form that facilitates the identification of similarities and differences between them. 
In particular, pre-defined associative operators (such as +,*,$\wedge$,$\vee$), whose explicit definitions are unfolded as part of our transformation, 
are always transformed into {\em right-associative} form (for example, $(x+y)+z$ is transformed to $x+(y+z)$).

\subsection{Embedding}

Generalisation is performed if the predicate obtained from distillation is an {\em embedding} of a previously distilled one. 
The form of embedding which we use to inform this process is known as {\em homeomorphic embedding}. 
The homeomorphic embedding relation was derived from results by Higman \cite{HIGMAN52} and Kruskal \cite{KRUSKAL60} 
and was defined within term rewriting systems \cite{DERSHOWITZ90} for detecting the possible divergence of the term rewriting process. 
Variants of this relation have been used to ensure termination within positive supercompilation \cite{SORENSEN94B}, 
distillation \cite{HAMILTON07A,HAMILTON12}, partial evaluation \cite{MARLET94} and partial deduction \cite{BOL93,LEUSCHEL98}.
\begin{definition}[Expression Embedding]
\normalfont{An expression $E$ is {\em embedded} in expression $E'$ if $E \trianglelefteq E'$, where the binary relation $\trianglelefteq$ is defined as follows.}
\end{definition}
\begin{center}
\begin{tabular}[t]{c@{\hspace*{1cm}}c@{\hspace*{1cm}}c}
$\infer{V \trianglelefteq V'}{}$ & $\infer{E \trianglelefteq \phi(E_1, \ldots, E_n)}{\exists i \in \{1 \ldots n\}.E \trianglelefteq E_i}$  & $\infer{\phi(E_1, \ldots, E_n) \trianglelefteq \phi(E_1', \ldots, E_n')}{\forall i \in \{1 \ldots n\}.E_i \trianglelefteq E_i'}$ 
\end{tabular}
\end{center}
The first rule here is for variables, the second is a {\em diving} rule and the third is a {\em coupling} rule.
Diving detects a sub-expression embedded in a larger expression, and coupling matches all the sub-expressions
of two expressions which have the same top-level functor. Bound variables are handled by this relation by requiring 
that they have the same de Bruijn indices. We write $E \preceq E'$ if expression $E$ is coupled with expression $E'$ at the top level.

\subsection{Generalisation}

\begin{definition}[Generalisation of Expressions]
\normalfont{The {\em generalisation} of expressions $E$ and $E'$ (denoted by $E~\sqcap~E'$) is defined as shown below.
\begin{center}
\begin{tabular}{l}
$E \sqcap E' = \left\{\begin{array}{ll}
(\phi(E_1'', \ldots, E_n''),\bigcup_{i=1}^{n} \theta_i,\bigcup_{i=1}^{n} \theta_i'), & $if $\phi = \phi' \\
\hspace*{0.5cm} $where$ \\
\hspace*{0.5cm} E = \phi(E_1, \ldots, E_n) \\
\hspace*{0.5cm} E' = \phi'(E_1', \ldots, E_n') \\
\hspace*{0.5cm} \forall i \in \{1 \ldots n\}.E_i \sqcap E_i' = (E_i'',\theta_i,\theta_i') \\
(V,\{V \mapsto E\},\{V \mapsto E'\}), & $otherwise $(V$ is fresh$)
\end{array}\right.$
\end{tabular}
\end{center}
The result of this generalisation is a triple $(E'',\theta,\theta')$ where $E''$ is the generalised expression and 
$\theta$ and $\theta'$ are substitutions s.t. $E'' \theta \equiv E$ and $E'' \theta' \equiv E'$.
Within these rules, if both expressions have the same functor at the outermost level, this is made the outermost functor 
of the resulting generalised expression, and the corresponding sub-expressions within the functor applications are then generalised. 
Otherwise, both expressions are replaced by the same variable.}
\end{definition}
\begin{definition}[Most Specific Generalisation]
\normalfont{A {\em most specific generalisation} of expressions $E$ and $E'$ is an expression $E''$ such
that for every other generalisation $E'''$ of $E$ and $E'$, there is a substitution $\theta$ such that $E''\theta \equiv E'''$.
The most specific generalisation, denoted by $E \triangle E'$, of expressions $E$ and $E'$ is computed by exhaustively 
applying the following rewrite rule to the triple obtained from the generalisation $E \sqcap E'$:
\begin{center}
$\left(\begin{array}{c}E, \\
\{V_1 \mapsto E', V_2 \mapsto E'\} \cup \theta, \\
\{V_1 \mapsto E'', V_2 \mapsto E''\} \cup \theta'
\end{array}\right)$
$\Rightarrow$
$\left(\begin{array}{c}E\{V_1 \mapsto V_2\}, \\
\{V_2 \mapsto E'\} \cup \theta, \\
\{V_2 \mapsto E''\} \cup \theta'
\end{array}\right)$
\end{center}
This minimises the substitutions by identifying common substitutions which were previously given different names.} 
\end{definition}

\section{Automatic Generation of Loop Invariants}

\subsection{Algorithm}

In order to calculate loop invariants, starting from the postcondition, we work our way backwards through each iteration of the loop
generating successive approximations to the loop invariant. If the current approximation is an embedding 
of a previous one (coupled at the top level) then these approximations are generalised with respect to each other. 
This process is continued until the current approximation is a renaming of a previous one; this is then the putative 
invariant for the loop. If there are a number of different possible paths through the loop, then a number of possible preconditions 
will be calculated for it; these are collapsed into a single precondition by being generalised with respect to each other, thus producing
a single approximation for each loop iteration.

Our algorithm for the automatic generation of an invariant for the loop 
$\mathsf{WHILE}~B~\mathsf{DO}~S$ with postcondition $Q$ is as shown in Figure \ref{algorithm}. 
\begin{figure}[htb]
\begin{center}
\begin{tabbing}
\hspace*{1cm} \=$f~(distill(\neg B \wedge Q))~\emptyset$ \\
\>{\bf where} \\
\>$f~P~\phi$ = \={\bf if}~$\exists Q \in \phi$ s.t. $Q \equiv P$ (modulo variable renaming) \\
\>\>{\bf then}~{\bf return}~$P$ \\
\>\>{\bf else}~\={\bf if}~$\exists Q \in \phi$ s.t. $Q \preceq P$ \\
\>\>\>{\bf then}~$f~P'~\phi$ where $P' = P \triangle Q$ \\
\>\>\>{\bf else}~\={\bf return}~$f~(\displaystyle \bigtriangleup_{i=1}^n \{distill(B \wedge P_i))~(\phi \cup \{P\})$ \\
\>\>\>\>{\bf where}~$WLP(S,P) = \displaystyle \bigwedge_{i=1}^n P_i$
\end{tabbing}
\end{center}
\caption{Algorithm for Finding Loop Invariants}
\label{algorithm}
\end{figure}
Here, $P$ is the current predicate (initially equivalent to $\neg B \wedge Q$, which is true if the loop is exited) and $\phi$ is the set of previous 
approximations to the invariant (initially empty). If $P$ is a renaming of a predicate in $\phi$, then $P$ is returned as the putative invariant. 
If there is a predicate $Q$ in $\phi$ which is embedded in $P$ (coupled at the top level), then $P$ and $Q$ are generalised with respect to 
each other, and the algorithm is further applied to this generalisation. Otherwise, the predicate which is true before the previous execution of 
the loop body is calculated. The conjuncts of this predicate are then combined with the loop condition (which must have been true for the loop 
body to be executed), simplified using distillation and generalised together. The algorithm is then further applied to the resulting generalised 
predicate with $P$ added to $\phi$.

The generated invariant may contain generalisation variables; inductive definitions for these variables can be determined 
using the three requirements for loop invariants (Definition \ref{invariant}). We try to find values for these variables that satisfy each of these 
requirements using our Poit\'{i}n theorem prover \cite{HAMILTON06}\footnote{This could also be done using a SAT solver.}. If we are not able to
satisfy all three of these requirements, then we have failed in finding a suitable invariant.

For the loop $\{P\}~\mathsf{WHILE} ~ B ~ \mathsf{DO} ~ \{I\} ~ S~\{Q\}$, the initial value of generalisation variable $v$ can be obtained 
by satisfying the following for $v_0$ using the first requirement:
\begin{center}
$P \Rightarrow I\{v:=v_0\}$
\end{center}
The inductive definition of $v$ can be obtained by satisfying the following for $v_{i+1}$ using the second requirement:
\begin{center}
$\{I\{v:=v_i\} \wedge B\}~S~\{I\{v:=v_{i+1}\}\}$
\end{center}
The final value of $v$ can be obtained by satisfying the following for $v_n$ using the third requirement:
\begin{center}
$(I\{v:=v_n\} \wedge \neg B) \Rightarrow Q$
\end{center}

\subsection{Example}

\setcounter{example}{0}
\begin{example}
\normalfont{We illustrate this algorithm by applying it to the example program in Figure \ref{example1}. \\
\\
Firstly, we calculate the logical assertion which is true if the loop is exited:
\begin{center}
$\mathsf{\neg(x < n) \wedge y = k^\wedge n}$
\end{center}
This is simplified by distillation to the following\footnote{For this and all following examples, the result of distillation is simplified by replacing any instances 
of the definitions of the pre-defined operators of our language with a corresponding call of the operator; the results would be too unwieldy otherwise.}:
\begin{equation}
\mathsf{x \geq n \wedge y = k^\wedge n}
\label{eq1}
\end{equation}
Then, we calculate the logical assertion which is true before the final execution of the loop body:
\begin{center}
$\mathsf{WLP(BEGIN~ x:=x+1; y:=y * k~END,x \geq n \wedge y = k^\wedge n)}$
\end{center}
This gives the following:
\begin{center}
$\mathsf{x+1 \geq n \wedge y * k = k^\wedge n}$
\end{center}
In conjunction with the loop condition ($\mathsf{x<n}$), this is simplified to the following by distillation:
\begin{equation}
\mathsf{x+1=n \wedge y * k = k^\wedge n}
\label{eq2}
\end{equation}
This is not an embedding of (\ref{eq1}), so the calculation continues.
We next calculate the logical assertion which is true before the penultimate execution of the loop body: 
\begin{center}
$\mathsf{WLP(BEGIN~ x:=x+1; y:=y * k~END,x+1 = n \wedge y * k = k^\wedge n)}$
\end{center}
This gives the following:
\begin{center}
$\mathsf{(x+1)+1=n \wedge (y * k) * k = k^\wedge n}$
\end{center}
In conjunction with the loop condition ($\mathsf{x<n}$), this is simplified to the following by distillation:
\begin{equation}
\mathsf{x+(1+1)=n \wedge y * (k * k) = k^\wedge n}
\label{eq3}
\end{equation}
We can see that (\ref{eq3}) is an embedding of (\ref{eq2}), so generalisation is performed to produce the following:
\begin{equation}
\mathsf{x+v = n \wedge y * w = k^\wedge n}
\label{eq4}
\end{equation}
This is not an embedding, so the logical assertion which is true before execution of the loop body is now re-calculated as follows:
\begin{center}
$\mathsf{WLP(BEGIN~ x:=x+1; y:=y * k~END,x+v = n \wedge y * w = k^\wedge n)}$
\end{center}
This gives the following:
\begin{center}
$\mathsf{(x+1)+v=n \wedge (y * k) * w = k^\wedge n}$
\end{center}
In conjunction with the loop condition ($\mathsf{x<n}$), this is simplified to the following by distillation:
\begin{equation}
\mathsf{x+(1+v)=n \wedge y * (k * w) = k^\wedge n}
\label{eq5}
\end{equation}
We can see that (\ref{eq5}) is an embedding of (\ref{eq4}), so generalisation is performed to produce the following:
\begin{equation}
\mathsf{x+v' = n \wedge y * w' = k^\wedge n}
\label{eq6}
\end{equation}
We can now see that (\ref{eq6}) is a renaming of (\ref{eq4}), so (\ref{eq6}) is our putative invariant
We now try to find inductive definitions for the generalisation variables $v'$ and $w'$ from the three requirements of loop invariants given in 
Definition \ref{invariant}, which we do using our theorem prover Poit\'{i}n.

The initial values of the generalisation variables, given by $v_0'$ and $w_0'$ can be determined using the first invariant requirement as follows:
\begin{center}
$\mathsf{n \geq 0 \wedge x = 0 \wedge y = 1 \Rightarrow x+v_0' = n \wedge y * w_0' = k^\wedge n}$
\end{center}
The assignments $v_0' := n$ and $w_0' := k^\wedge n$ satisfy this assertion.

The inductive values of the generalisation variables, given by $v_{i+1}'$ and $w_{i+1}'$, can be determined using the second invariant requirement as follows:
\begin{center}
$\mathsf{x+v_i' = n \wedge y * w_i' = k^\wedge n \wedge x < n \Rightarrow (x+1)+v_{i+1}' = n \wedge (y*k) * w_{i+1}' = k^\wedge n}$
\end{center}
The assignments $v_{i+1}' := v_i' - 1$ and $w_{i+1}' := w_i'/k$ satisfy this assertion.

The final values of the generalisation variables, given by $v_n'$ and $w_n'$, can be determined using the third invariant requirement as follows:
\begin{center}
$\mathsf{x+v_n = n \wedge y * w_n = k^\wedge n \wedge \neg(x < n) \Rightarrow y = k^\wedge n}$
\end{center}
The assignments $v_n := 0$ and $w_n = 1$ satisfy this assertion.
The discovered invariant is therefore equivalent to the following\footnote{Our technique does not actually convert the discovered invariant into this
simplified form; however the discovered inductive version is sufficient for the purpose of verifying the program.}:
\begin{center}
$\mathsf{x \leq n \wedge y = k^\wedge x}$
\end{center}
}
\end{example}

\subsection{Termination}

In order to prove that our loop invariant algorithm always terminates, we firstly need to show that in any infinite sequence of predicates $P_0, P_1, \ldots$ 
there definitely exists some $i < j$ where $P_i \preceq P_j$. This amounts to proving that the embedding relation $\preceq$ is a {\em well-quasi order}.
\begin{definition}[Well-Quasi Order]
\normalfont{A well-quasi order on set $S$ is a reflexive, transitive relation $\leq$ such that for
any infinite sequence $s_1, s_2, \ldots$ of elements from $S$ there are numbers $i, j$ with
$i < j$ and $s_i \leq s_j$.}
\end{definition}
\begin{lemma}[$\preceq$ is a Well-Quasi Order]
\normalfont{The embedding relation $\preceq$ is a well-quasi order on any sequence of predicates.}
\label{wellquasi}
\end{lemma}
\begin{proof}
\normalfont{The proof is similar to that given in \cite{KLYUCHNIKOV10A}. It involves showing that there are a
finite number of functors (function names and constructors) in the language. Applications of different arities are 
replaced with separate constructors; we prove that arities are bounded so there are a finite number of these. 
We also replace case expressions with constructors. Since bound variables are defined using de Bruijn indices,
each of these are replaced with separate constructors; we also prove that de Bruijn indices are bounded.
The overall number of functors is therefore finite, so Kruskal's tree theorem can then be applied to show that 
$\preceq$ is a well-quasi-order.}
\end{proof}
\begin{theorem}[Termination]
\normalfont{The loop invariant algorithm always terminates.}
\end{theorem}
\begin{proof}
\normalfont{
The proof is by contradiction. If the loop invariant algorithm did not terminate then the set of invariants generated as successive approximations
to the invariant must be infinite. Every new predicate which is added to the set of approximations cannot have 
any of the previously generated predicates on this set embedded within it by the homeomorphic embedding relation $\preceq$, since either 
generalisation would have been performed or a renaming encountered and the algorithm terminated. However, this contradicts the 
fact that $\preceq$ is a well-quasi-order (Lemma \ref{wellquasi}).}
\end{proof} 
\section{Further Examples}

In this section, we consider further examples that have been found to be problematic in techniques (including our own) for the automatic discovery of 
invariants.
\begin{example}
\normalfont{Consider the following example program: \\
\\
$~~~\{\mathsf{n \geq 0}\}$ \\
$~~~\mathsf{x:=n;}$ \\
$~~~\mathsf{y:=1;}$ \\
$~~~\mathsf{z:=k;}$ \\
$~~~\mathsf{WHILE ~ x > 0 ~ DO}$ \\
$~~~~~~\mathsf{BEGIN}$ \\
$~~~~~~~~~ \mathsf{IF}~x\%2=1~\mathsf{THEN}~y:=y*z~\mathsf{ELSE}~\mathsf{SKIP};$ \\
$~~~~~~~~~ \mathsf{x:=x/2;}$ \\
$~~~~~~~~~ \mathsf{z:=z*z}$ \\
$~~~~~~ \mathsf{END}$ \\
$~~~\{\mathsf{y = k^\wedge n}\}$ \\
\\
This program also calculates the exponentiation $\mathsf{k^\wedge n}$. This example is problematic using other approaches (as discussed in Section \ref{related}) 
because of the presence of the conditional inside the loop, which causes an exponential blow-up in the size of the generated predicates; we show how this 
blow-up is avoided using our approach. In the following, we use $S$ to denote the body of the loop in the above program. Firstly, we calculate the logical 
assertion which is true if the loop is exited:
\begin{center}
$\mathsf{\neg(x > 0) \wedge y = k^\wedge n}$
\end{center}
This is simplified by distillation to the following:
\begin{equation}
\mathsf{x \leq 0 \wedge y = k^\wedge n}
\label{eq7}
\end{equation}
Then, we calculate the logical assertion which is true before the final execution of the loop body:
\begin{center}
$\mathsf{WLP(S,x \leq 0 \wedge y = k^\wedge n)}$
\end{center}
This gives the following:
\begin{center}
$\mathsf{(x\%2=1 \Rightarrow x/2 \leq 0 \wedge y*z=k^\wedge n) \wedge (\neg(x\%2=1) \Rightarrow x/2 \leq 0 \wedge y=k^\wedge n)}$
\end{center}
In conjunction with the loop condition ($\mathsf{x > 0}$), the second conjunct is simplified to $\mathsf{True}$ by distillation, but the first conjunct is simplified to the following:
\begin{equation}
\mathsf{x=1 \wedge y*z = k^\wedge n}
\label{eq8}
\end{equation}
This is not an embedding, so the calculation continues.
We next calculate the logical assertion which is true before the penultimate execution of the loop body: 
\begin{center}
$\mathsf{WLP(S,x=1 \wedge y*z = k^\wedge n)}$
\end{center}
This gives the following:
\begin{center}
$\mathsf{(x\%2=1 \Rightarrow x/2 =1 \wedge (y*z)*(z*z)=k^\wedge n) \wedge (\neg(x\%2=1) \Rightarrow x/2 =1 \wedge y*(z*z)=k^\wedge n)}$
\end{center}
In conjunction with the loop condition ($\mathsf{x > 0}$), the first conjunct is simplified to the following by distillation:
\begin{center}
$\mathsf{x=(2*1)+1 \wedge y*(z*(z*z))=k^\wedge n}$
\end{center}
and the second conjunct is simplified to the following:
\begin{center}
$\mathsf{x=2*1 \wedge y*(z*z)=k^\wedge n}$
\end{center}
These are generalised with respect to each other to give the following:
\begin{equation}
\mathsf{x=v \wedge y*(z*w)=k^\wedge n} 
\label{eq9}
\end{equation}
This is not an embedding, so the logical assertion which is true before execution of the loop body is now re-calculated as follows:
\begin{center}
$\mathsf{WLP(S,x=v \wedge y*(z*w) = k^\wedge n)}$
\end{center}
This gives the following:
\begin{center}
$\mathsf{(x\%2=1 \Rightarrow x/2 = v \wedge (y*z)*((z*z)*w)=k^\wedge n)~\wedge}$ \\ 
$\mathsf{(\neg(x\%2=1) \Rightarrow x/2 = v \wedge y*((z*z)*w)=k^\wedge n)}$
\end{center}
In conjunction with the loop condition ($\mathsf{x > 0}$), the first conjunct is simplified to the following by distillation:
\begin{center}
$\mathsf{x=(2*v)+1 \wedge y*(z*(z*(z*w)))=k^\wedge n}$
\end{center}
and the second conjunct is simplified to the following:
\begin{center}
$\mathsf{x=2*v \wedge y*(z*(z*w))=k^\wedge n}$
\end{center}
These are generalised with respect to each other to give the following:
\begin{equation}
\mathsf{x=v' \wedge y*(z*(z*w'))=k^\wedge n} 
\label{eq10}
\end{equation}
We can see that (\ref{eq10}) is an embedding of (\ref{eq9}), so generalisation is performed to give the following:
\begin{equation}
\mathsf{x=v'' \wedge y*(z*w'')=k^\wedge n}
\label{eq11}
\end{equation}
We can now see that (\ref{eq11}) is a renaming of (\ref{eq9}), so (\ref{eq11}) is our putative invariant. 
We now try to find inductive definitions for the generalisation variables $v''$ and $w''$ from the three requirements of loop invariants given in 
Definition \ref{invariant}, which we do using our theorem prover Poit\'{i}n.

The initial values of the generalisation variables, given by $v_0''$ and $w_0''$ can be determined using the first invariant requirement as follows:
\begin{center}
$\mathsf{n \geq 0 \wedge x = n \wedge y = 1 \wedge z = k \Rightarrow x = v_0'' \wedge y * (z * w_0'') = k^\wedge n}$
\end{center}
The assignments $v_0'' := n$ and $w_0'' := k^\wedge (n-1)$ satisfy this assertion.

The inductive values of the generalisation variables, given by $v_{i+1}''$ and $w_{i+1}''$, can be determined using the second invariant requirement as follows:
\begin{center}
$\mathsf{x=v_i'' \wedge y * (z * w_i'') = k^\wedge n \wedge x >0 \Rightarrow}$ \\
$\mathsf{((x\%2=0 \Rightarrow x/2=v_{i+1}'' \wedge (y*z) * ((z * z) * w_{i+1}'') = k^\wedge n)~\wedge}$ \\
$\mathsf{(\neg(x\%2=1) \Rightarrow x/2=v_{i+1}'' \wedge y * ((z * z) * w_{i+1}'') = k^\wedge n))}$
\end{center}
If $x\%2=0$, the assignments $v_{i+1}'' := v_i''/2$ and $w_{i+1}'' := w_i''/(z*z)$ satisfy this assertion. Otherwise, the assignments $v_{i+1}'' := v_i''/2$ and $w_{i+1}'' := w_i''/z$ satisfy this assertion.

The final values of the generalisation variables, given by $v_n''$ and $w_n''$, can be determined using the third invariant requirement as follows:
\begin{center}
$\mathsf{x=v_n'' \wedge y * (z*w_n'') = k^\wedge n \wedge \neg(x > 0) \Rightarrow y = k^\wedge n}$
\end{center}
The assignments $v_n'' := 0$ and $w_n'' = 1/z$ satisfy this assertion.
The discovered invariant is therefore equivalent to the following (again, we do not actually convert the invariant into this simplified form):
\begin{center}
$\mathsf{y*z^\wedge x = k^\wedge n}$
\end{center}
}
\end{example}
\begin{example}
\normalfont{Consider the following example program: \\
\\
$~~~\{\mathsf{n \geq 0}\}$ \\
$~~~\mathsf{x:=0;}$ \\
$~~~\mathsf{y:=1;}$ \\
$~~~\mathsf{WHILE ~ x < n ~ DO}$ \\
$~~~~~~\mathsf{BEGIN}$ \\
$~~~~~~~~~\mathsf{x:=x + 1;}$ \\
$~~~~~~~~~\mathsf{z:=0;}$ \\
$~~~~~~~~~\mathsf{v:=0;}$ \\
$~~~~~~~~~\mathsf{WHILE ~ z < k ~ DO}$ \\
$~~~~~~~~~~~~\mathsf{BEGIN}$ \\
$~~~~~~~~~~~~~~~\mathsf{v:=v + y;}$ \\
$~~~~~~~~~~~~~~~\mathsf{z:=z + 1}$ \\
$~~~~~~~~~~~~ \mathsf{END;}$ \\
$~~~~~~~~~\mathsf{y:=v}$ \\
$~~~~~~ \mathsf{END}$ \\
$~~~\{\mathsf{y = k^\wedge n}\}$ \\
\\
This program also calculates the exponentiation $\mathsf{k^\wedge n}$, but can be problematic using other approaches (as discussed in Section \ref{related})
because it uses a nested loop.
We will assume that the programmer has given the postcondition of the inner loop as $\{\mathsf{v = y * k}\}$. Note that our
technique could also be applied without this information, as the postcondition from the weakest liberal precondition calculation for the 
outer loop body could be used instead. However, the invariant of the inner loop would have to be re-calculated for every weakest liberal 
precondition calculation of the outer loop body. The logical assertion which is true if the inner loop is exited is as follows: 
\begin{center}
$\mathsf{\neg (z < k) \wedge v = y * k}$
\end{center}
This is simplified by distillation to the following:
\begin{equation}
\mathsf{z \geq k \wedge v = y * k}
\label{eq12}
\end{equation}
Then, we calculate the logical assertion which is true before the final execution of the inner loop body: 
\begin{center}
$\mathsf{WLP(BEGIN~v:=v+y; z:=z+1~END,z \geq k \wedge v = y * k)}$
\end{center}
This gives the following:
\begin{center}
$\mathsf{z+1 \geq k \wedge v+y = y * k}$
\end{center}
In conjunction with the loop condition ($\mathsf{z < k}$), this is simplified to the following by distillation:
\begin{equation}
\mathsf{z+1=k \wedge v+y = y * k}
\label{eq13}
\end{equation}
This is not an embedding, so the calculation continues.
The logical assertion which is true before the penultimate execution of the inner loop body is as follows: 
\begin{center}
$\mathsf{WLP(BEGIN~v:=v+y; z:=z+1~END,z+1 = k \wedge v+y = y * k)}$
\end{center}
This gives the following:
\begin{center}
$\mathsf{(z+1)+1 = k \wedge (v+y)+y = y * k}$
\end{center}
In conjunction with the loop condition ($\mathsf{z < k}$), this is simplified to the following by distillation:
\begin{equation}
\mathsf{z+(1+1)=k \wedge v+(y+y) = y * k}
\label{eq14}
\end{equation}
We can see that (\ref{eq14}) is an embedding of (\ref{eq13}), so generalisation is performed to produce the following:
\begin{equation}
\mathsf{z+w=k \wedge v+u = y * k}
\label{eq15}
\end{equation}
This is not an embedding, so the calculation continues.
We next calculate the logical assertion which is true before the penultimate execution of the inner loop body: 
\begin{center}
$\mathsf{WLP(BEGIN~v:=v+y; z:=z+1~END,z+w = k \wedge v+u = y * k)}$
\end{center}
This gives the following:
\begin{center}
$\mathsf{(z+1)+w = k \wedge (v+y)+u = y * k}$
\end{center}
In conjunction with the loop condition ($\mathsf{z < k}$), this is simplified to the following by distillation:
\begin{equation}
\mathsf{z+(1+w)=k \wedge v+(y+u) = y * k}
\label{eq16}
\end{equation}
We can now see that (\ref{eq16}) is an embedding of (\ref{eq15}), so generalisation is performed to produce the following:
\begin{equation}
\mathsf{z+w'=k \wedge v+u' = y * k}
\label{eq17}
\end{equation}
We can now see that (\ref{eq17}) is a renaming of (\ref{eq15}), so (\ref{eq17}) is our putative invariant.
We now try to find inductive definitions for the generalisation variables $w'$ and $u'$ from the three requirements of loop invariants given in 
Definition \ref{invariant}, which we do using our theorem prover Poit\'{i}n.

The initial values of the generalisation variables, given by $w_0'$ and $u_0'$ can be determined using the first invariant requirement as follows:
\begin{center}
$\mathsf{n \geq 0 \wedge x \leq n \wedge z = 0 \wedge v = 0 \Rightarrow z+w_0'=k \wedge v+u_0' = y * k}$
\end{center}
The assignments $w_0' := k$ and $u_0' := y * k$ satisfy this assertion.

The inductive values of the generalisation variables, given by $w_i'$ and $u_i'$, can be determined using the second invariant requirement as follows:
\begin{center}
$\mathsf{z+w_i'=k \wedge v+u_i' = y * k \wedge z < k \Rightarrow (z+1)+w_{i+1}'=k \wedge (v+y)+u_{i+1}' = y * k}$
\end{center}
The assignments $w_{i+1}' := w_i' - 1$ and $u_{i+1}' := u_i'-y$ satisfy this assertion.

The final values of the generalisation variables, given by $w_n'$ and $u_n'$, can be determined using the third invariant requirement as follows:
\begin{center}
$\mathsf{z+w_n'=k \wedge v+u_n' = y * k \wedge \neg(z < k) \Rightarrow v = y * k}$
\end{center}
The assignments $w_n' := 0$ and $u_n' = 0$ satisfy this assertion.
This is equivalent to the following:
\begin{center}
$\mathsf{z \leq k \wedge v = y * z}$
\end{center}
This invariant can then be used to calculate the invariant for the outer loop as shown in Example 1.}
\end{example}
\begin{example}
\normalfont{Consider the following example program: \\
\\
$~~~\{\mathsf{n \geq 0}\}$ \\
$~~~\mathsf{x:=0;}$ \\
$~~~\mathsf{y:=1;}$ \\
$~~~\mathsf{WHILE ~ x < n ~ DO}$ \\
$~~~~~~\mathsf{BEGIN}$ \\
$~~~~~~~~~ \mathsf{x:=x + 1;}$ \\
$~~~~~~~~~ \mathsf{y:=k * y}$ \\
$~~~~~~ \mathsf{END}$ \\
$~~~\{\mathsf{y = k^\wedge n}\}$ \\
\\
This program is very similar to that shown in Figure \ref{example1}, except that the operands of the final multiplication are swapped. 
We now show how this program is problematic using our approach. 
Firstly, we calculate the logical assertion which is true if the loop is exited:
\begin{center}
$\mathsf{\neg(x < n) \wedge y = k^\wedge n}$
\end{center}
This is simplified by distillation to the following:
\begin{equation}
\mathsf{x \geq n \wedge y = k^\wedge n}
\label{eq18}
\end{equation}
Then, we calculate the logical assertion which is true before the final execution of the loop body:
\begin{center}
$\mathsf{WLP(BEGIN~ x:=x+1; y:=k * y~END,x \geq n \wedge y = k^\wedge n)}$
\end{center}
This gives the following:
\begin{center}
$\mathsf{x+1 \geq n \wedge k * y = k^\wedge n}$
\end{center}
In conjunction with the loop condition ($\mathsf{x<n}$), this is simplified to the following by distillation:
\begin{equation}
\mathsf{x+1=n \wedge k * y = k^\wedge n}
\label{eq19}
\end{equation}
This is not an embedding, so the calculation continues.
We next calculate the logical assertion which is true before the penultimate execution of the loop body: 
\begin{center}
$\mathsf{WLP(BEGIN~ x:=x+1; y:=k * y~END,x+1 = n \wedge k * y = k^\wedge n)}$
\end{center}
This gives the following:
\begin{center}
$\mathsf{(x+1)+1=n \wedge k * (k * y) = k^\wedge n}$
\end{center}
In conjunction with the loop condition ($\mathsf{x<n}$), this is simplified to the following by distillation:
\begin{equation}
\mathsf{x+(1+1)=n \wedge k * (k * y) = k^\wedge n}
\label{eq20}
\end{equation}
We can see that (\ref{eq20}) is an embedding of (\ref{eq19}), so generalisation is performed to produce the following:
\begin{equation}
\mathsf{x+v = n \wedge k * w = k^\wedge n}
\label{eq21}
\end{equation}
This is not an embedding, so the logical assertion which is true before execution of the loop body is now re-calculated as follows: 
\begin{center}
$\mathsf{WLP(BEGIN~ x:=x+1; y:=y * k~END,x+v = n \wedge k * w = k^\wedge n)}$
\end{center}
This gives the following:
\begin{center}
$\mathsf{(x+1)+v = n \wedge k * w = k^\wedge n}$
\end{center}
In conjunction with the loop condition ($\mathsf{x<n}$), this is simplified to the following by distillation:
\begin{equation}
\mathsf{x+(1+v) = n \wedge k * w = k^\wedge n}
\label{eq22}
\end{equation}
We can see that (\ref{eq22}) is an embedding of (\ref{eq21}), so generalisation is performed to produce the following:
\begin{equation}
\mathsf{x+v' = n \wedge k * w = k^\wedge n}
\label{eq23}
\end{equation}
We can now see that (\ref{eq23}) is a renaming of (\ref{eq21}), so (\ref{eq23}) is our putative invariant.
Using our theorem prover Poit\'{i}n, we are unable to prove that this invariant satisfies any of the three requirements for the loop invariant given in Definition \ref{invariant}. The problem here is that a variable that is updated in the loop body (in this case $y$) has been removed from the calculated invariant by generalisation. This situation can be avoided to a certain extent by making sure that such variables appear in the left operand of binary operations (as is
the case in Example 1). This is because repeated applications of such operations are always transformed into right-associative form by distillation,
and any mismatches are more likely to occur in the right operand. However, this may not always be possible if the operation is not commutative or if
both operands contain variables that are updated in the loop body.
}
\end{example}

\section{Related Work}
\label{related}
\setcounter{equation}{0}

The main approaches to the automatic generation of loop invariants include abstract interpretation, proof planning, dynamic methods, 
using heuristics and the induction-iteration method. The earliest methods for the automatic generation of loop invariants involved static analysis. 
{\em Abstract interpretation} is a symbolic execution of programs over abstract domains (such as predicate abstraction domains or polyhedral abstraction domains) that over-approximates the semantics of loop iteration. {\em Predicate abstraction} domains \cite{AGERWALA78, GRAF97, SAIDI99, DAS99, 
FLANAGAN02} replace predicates with variables, which is similar to the generalisation we perform in our approach. 
Constraint-based techniques rely on sophisticated decision procedures over non-trivial mathematical domains (such as polynomials 
\cite{SANKARANARAYANAN04} or convex polyhedra \cite{COUSOT78}) to represent concisely the semantics of loops with respect to certain properties. 
Loop invariants in these forms are extremely useful but rarely sufficient to prove full functional correctness of programs.

In \cite{FLANAGAN02}, Flanagan and Qadeer describe the use of predicate abstraction to generate loop invariants. Their
approach differs from our own in that predicates are obtained by working forwards from the precondition through successive 
iterations of the loop, as opposed to backwards from the postcondition in our approach. A {\em strongest postcondition semantics}
is therefore used in \cite{FLANAGAN02} as opposed to our weakest precondition approach. Loop invariants are computed by 
iterative approximation. The first approximation is obtained by abstracting the set of reachable states at loop entry. Each 
successive approximation enlarges the current approximation to include the states reachable by executing the loop body once 
from the states in the current approximation. The iteration terminates in a loop invariant since the abstract domain is finite. 
However, this approach does suffer from the drawback that the approximations can grow exponentially as they are a disjunction of
the approximations for all the reachable states. This exponential growth is avoided in our approach. Also, we argue that working 
forwards from the precondition makes it harder to find the required invariant since (as observed in \cite{FURIA10}), 
the required invariant is often a weakening of the postcondition. 

In \cite{IRELAND97}, a {\em proof planning} approach is used to synthesise loop invariants. This approach makes use of failed attempts
to prove a putative invariant correct. The proof attempts are applied to the verification conditions generated for the putative invariant.
If these proof attempts fail, the failure is analysed using {\em proof critics}. One such critic is the generalisation critic, which performs
generalisation in a similar way to that described in our work, and is used to update the putative invariant to one which is more likely to
be correct. One drawback of this approach is that the original putative invariant has to be guessed, although the postcondition is a
good first guess. Another drawback is knowing which critics to apply when, since multiple critics can discover the invariant, but some
may do so more efficiently than others. Also, it is not clear how this method could be applied to nested loops.

In \cite{ERNST01}, invariants are discovered dynamically. Using this approach, the program is run over a test suite of inputs.
The corresponding outputs are analysed for patterns and relationships among the variables. Candidate invariants are guessed by trying out a 
pre-defined set of user-provided templates (including comparisons between variables, simple inequalities, and simple list comprehensions). 
These candidate invariants are then tested against several program runs; the invariants that are not violated in any of the runs are retained 
as likely invariants. This inference is not sound and only gives an educated guess. However, a prototype tool called Daikon was implemented 
using these techniques, and has worked well in practice and many of the guessed invariants are sound.

In \cite{FURIA10}, Furia and Meyer describe the use of {\em heuristics} to synthesise loop invariants. This work is based on
the observation that the required invariant is often a weakening of the postcondition for the loop and can be obtained 
by mutating this postcondition. The core idea is to generate candidate invariants by mutating postconditions according 
to a few commonly recurring patterns. Although this idea works well in many cases, it is not capable of generating the 
required invariant for the example program in Figure \ref{example1}, as this requires the addition of an extra clause to
the postcondition ($\mathsf{y \leq n}$), which is not one of the described heuristics. 

The previous work which is closest to our own is the {\em induction-iteration} method of Suzuki and Ishihata \cite{SUZUKI77}. This method
works as follows for the program $\mathsf{\{P\}~S_1;WHILE~B~DO~S_2~\{Q\}}$ with precondition $P$ and postcondition $Q$. 
Firstly, the logical assertion which is true if the loop is exited is calculated in a similar way to our technique:
\begin{center}
$P_0 = (\neg B \Rightarrow Q)$
\end{center}
Then, similarly to our technique, the weakest liberal precondition is used to calculate the logical assertion which is true before each execution 
of the loop body (in reverse order):
\begin{center}
\begin{tabular}{l}
$P_{i+1} = (B \Rightarrow WLP(S_2,P_i))$
\end{tabular}
\end{center}
The weakest liberal precondition of the loop is given by $\bigwedge\limits_{i=0}^{\infty}P_i$. In order to calculate this finitely, a number of successive 
approximations are calculated for it until one is found that is a loop invariant, where the $j^{th}$ approximation is given by $I_j = \bigwedge\limits_{i=0}^{j}P_i$. 
It then has to be shown that this approximation is true on entry to the loop and is also a loop invariant:
\begin{equation}
P \Rightarrow WLP(S_1,I_j)
\label{e1}
\end{equation}
\begin{equation}
(I_j \wedge B) \Rightarrow WLP(S_2,I_j)
\label{e2}
\end{equation}
(\ref{e2}) is equivalent to the following:
\begin{equation}
I_j \Rightarrow P_{j+1}
\label{e3}
\end{equation}
This therefore suggests an iterative approach to finding the loop invariant. Successive values for $I_j$ can be computed making use of the
previous values. If (\ref{e3}) is satisfied for the current value of $I_j$, then we are done and $I_j$ is the required invariant.
If (\ref{e1}) is not satisfied for the current value of $I_j$, then we have failed to find a suitable invariant. Otherwise, we carry on
the iteration to $I_{j+1}$.

One problem with this approach is that, unlike our own approach, it is not guaranteed to terminate. This is avoided by limiting the number of iterations. 
It is found that  in practice, for most of the small examples tried, very few iterations are actually required. 
Another problem with this approach is that there can be an exponential blow-up in clauses into increasingly larger conjunctions. 
This is particularly the case for conditionals, and can degrade $I_j$ to such an extent that it never converges to a loop invariant. 
This problem is avoided in \cite{SUZUKI77} by cleverly designing the theorem prover to avoid this potential exponential blow-up in clauses. 
In our approach, this problem is avoided by combining the conjuncts using generalisation. Finally, the seminal work in \cite{SUZUKI77} 
does not show how to deal with nested loops. However, an extension to the technique which does this is described by Xu et al. \cite{XU00}. 
This approach is very similar to the way in which we also deal with nested loops.

\section{Conclusions and Further Work}

In this paper we have described a technique for automatically discovering loop invariants. 
The technique we describe is similar to the induction-iteration method of Suzuki and Ishihata \cite{SUZUKI77}, 
but we overcome the problems associated with that method. One of these problems was the potential non-termination of the induction-iteration method;
our technique is guaranteed to terminate but may not be able to find a suitable invariant. Another problem with the induction-iteration method is the potential exponential blow-up in clauses into
increasingly larger conjunctions. Our technique avoids this through the combination of these conjuncts using generalisation. We have successfully 
demonstrated our technique on example imperative programs that have proven to be problematic using other approaches. We have also characterised
the situations where our technique fails to find an invariant and shown how this can be ameliorated to a certain extent.

There are a number of possible directions for further work. Firstly, we need to extend our techniques to languages with richer features.
For example, we could extend the language to manipulate unbounded data structures such as arrays. For such constructs, the required
loop invariants need to be universally quantified, but this can be handled by our theorem prover Poit\'{i}n, so should not present a problem 
for our technique. Another way in which the language could be extended would be to handle pointers. Separation logic \cite{OHEARN01,REYNOLDS02}
extends Floyd-Hoare logic to be able to handle pointers, so this seems to be an obvious basis for the extension of our technique. It has already been
shown by Ireland \cite{IRELAND06} how his approach to invariant generation can be extended to handle pointers by making use of separation logic.

One other possible direction for further work is extending our technique to deal with the termination of programs. This would involve calculating
the {\em weakest precondition} rather than the weakest liberal precondition as we do here. This would require the generation of a {\em variant}
in addition to an invariant, and the refinement of the invariant to show that the variant is decreased on each iteration of the loop. This appears to be
a lot more challenging than the problem which is tackled here.

\bibliographystyle{eptcs}

\bibliography{mybib}

\end{document}